\documentclass[12pt]{amsart}
\usepackage{geometry}                
\usepackage{graphicx,mathrsfs,epstopdf}
\usepackage{amsmath,amsfonts,amssymb,amsthm}
\usepackage[colorlinks,citecolor=blue,urlcolor=blue]{hyperref}
\usepackage{tikz}
\usepackage{tikz-cd}
\usepackage{booktabs}
\usepackage{multirow}
\usepackage{lscape}

\newtheorem{thm}{Theorem}[section]
 
\newtheorem{prop}[thm]{Proposition}

\theoremstyle{definition}

\newtheorem{conj}[thm]{Conjecture} 
\newtheorem{prob}[thm]{Problem}
\newtheorem{exmp}[thm]{Example}

\theoremstyle{remark}

\newcommand{\QED}{\ifhmode\unskip\nobreak\fi\quad {\rm Q.E.D.}} 

\newcommand{\R}{\mathbb{R}}
\renewcommand{\S}{\mathbb{S}}

\newcommand{\bs}{\boldsymbol}

\newcommand{\de}{{\rm de}}

\newcommand{\julia}{{\bf julia}$\mathbf{>}$\;}

\definecolor{sascha_color}{RGB}{44,227,147}

\title[Linear covariance models with numerical algebra]{Estimating linear covariance models
\\ with numerical nonlinear algebra}
\author[]{Bernd Sturmfels}
\address{Max Planck-Institute for Mathematics in the Sciences, Leipzig, Germany}
\email{bernd@mis.mpg.de}

\author[]{Sascha Timme}
\address{Institute of Mathematics, Technische Universit\"{a}t Berlin, Berlin,  Germany}
\email{timme@math.tu-berlin.de}

\author[]{Piotr Zwiernik}
\address{Dept.~of Economics and Business, Universitat Pompeu Fabra, Barcelona, Spain}
\email{piotr.zwiernik@upf.edu}

\begin{document}
\begin{abstract}
Numerical nonlinear algebra is applied to
maximum likelihood estimation 
for Gaussian models defined by linear constraints on
the covariance matrix. 
We examine the generic case as well as 
special models (e.g.~Toeplitz, sparse, trees)
that are of interest in statistics.
We study the maximum likelihood degree
and its dual analogue, and we introduce a new software package
{\tt LinearCovarianceModels.jl}
for solving the score equations.
All local maxima can thus be computed reliably. In addition we identify several scenarios 
for which the estimator is a rational function. 
\end{abstract}

\maketitle

\section{Introduction}

In many statistical applications, the covariance matrix $\Sigma$ has a special structure.
A natural setting is that one imposes linear constraints on $\Sigma$ or its inverse $\Sigma^{-1}$.
 We here study models for Gaussians whose covariance matrix~$\Sigma$
 lies in a given linear~space.
  Such linear Gaussian covariance models were introduced by Anderson~\cite{andersonLinearCovariance}.
  He was motivated by the Toeplitz structure of $\Sigma$ in time series analysis.
  Recent applications of such models include repeated time series, longitudinal data, and a range of
     engineering problems~\cite{pourahmadi1999joint}. 
    Other occurrences are Brownian motion tree models 
\cite{sturmfels2019brownian}, as well as pairwise independence models,
where some entries of $\Sigma$ are set to zero.

We are interested in maximum likelihood estimation (MLE) for linear covariance models.
This is a nonlinear algebraic optimization problem over a spectrahedral cone,
namely the convex cone of positive definite matrices $\Sigma$ satisfying the constraints.
The objective function is not convex and can 
have multiple local maxima. Yet, if the sample size is large relative to the dimension,
then the problem is essentially convex. This was shown in \cite{ZUR14}. In general,
however, the MLE problem is poorly understood, and there is a need for
accurate methods that reliably identify all local maxima.

Nonlinear algebra furnishes such a method, namely solving the
score equations \cite[Section 7.1]{Sul} using numerical homotopy continuation \cite{SW}.
This is guaranteed to find all critical points of the likelihood function
and hence all local maxima. A key step is the knowledge of the
maximum likelihood degree (ML degree). This is the number 
of complex critical points. The ML degree of a linear covariance model is an
 invariant of a linear space of symmetric matrices which is of interest in its own right.

Our presentation is organized as follows. In Section~\ref{sec2} we introduce
various models to be studied, ranging from generic linear equations to
colored graph models. In Section~\ref{sec3} we discuss the maximum
likelihood estimator as well as the dual maximum likelihood estimator.
Starting from \cite[Proposition 7.1.10]{Sul}, we derive a convenient form of the 
score equations. The natural point of entry for an algebraic geometer is the study of
generic linear constraints. This is our topic in Section~\ref{sec4}.
We compute a range of ML degrees and we compare them to the
dual degrees  in \cite[Section 2.2]{SU}.

In Section~\ref{sec5} we present our software {\tt LinearCovarianceModels.jl}.
This is written in  {\tt Julia} and it is easy to use. It
computes the ML degree and the dual ML degree
for a given subspace $\mathcal{L}$, and it determines all complex critical points for a given
sample covariance matrix $S$. Among these, it identifies the real and positive definite solutions,
and it then selects those that are local maxima. The package is available~at
\begin{equation}
\label{eq:uurrll}
 \hbox{
	\url{https://github.com/saschatimme/LinearCovarianceModels.jl}
	}
\end{equation}
This rests on the software {\tt HomotopyContinuation.jl} due to Breiding and Timme~\cite{BT}.

Section~\ref{sec6} discusses instances where
the likelihood function has multiple local maxima. This is
meant to underscore the strength of our approach.
We then turn to models where the maximum is unique and the MLE is a rational function.

In Section~\ref{sec7} we examine Brownian motion tree models. Here the
linear constraints are determined by a rooted phylogenetic tree. We study 
the ML degree and dual ML degree.  We show that the latter equals one
for binary trees, and we derive the explicit rational formula for their MLE.
A census of these degrees is found in Table~\ref{tab:BMT}.

\section{Models}\label{sec2}

Let $\mathbb S^n$ be the $\binom{n+1}{2}$-dimensional real
vector space of $n\times n$ symmetric matrices $\Sigma = (\sigma_{ij})$. 
The subset $\mathbb S^n_+$ of positive definite matrices is a full-dimensional
 open convex cone. 
 Consider any linear subspace $\mathcal L$ of $\mathbb S^n$ whose 
  intersection with $\mathbb S^n_+$ is nonempty. Then
  $\mathbb S^n_+\cap  \mathcal L$ is a relatively open
  convex cone. In optimization, where one uses the closure,
  this is known as    a {\em spectrahedral cone}. In statistics, the intersection
$  \mathbb S^n_+ \cap  \mathcal L$ is   a {\em linear covariance model}.
These are the models we study in this paper.
In what follows we discuss various families of linear spaces $\mathcal{L}$
that are of interest to us.

\smallskip

\textbf{Generic linear constraints:}
Fix a positive integer $ m \le \binom{n+1}{2}$ and suppose that
$\mathcal{L}$ is a generic linear subspace of $\mathbb S^n$. 
Here ``generic'' is meant in the sense of algebraic geometry, i.e.~$\mathcal{L}$
is a point in the Grassmannian that lies outside a certain algebraic hypersurface.
This hypersurface has measure zero, so a random subspace 
will be generic with probability one. For a geometer, it is natural to begin
with the generic case, since its complexity controls the complexity of any
special family of linear spaces. In particular, the ML degree for a generic
$\mathcal{L}$ depends only on $m$ and $n$, and this furnishes an upper bound
for the ML degree of the special families below.

\smallskip

\textbf{Diagonal covariance matrices:}
Here we take $m \le n$ and we assume that $\mathcal{L}$ is a linear
space that consists of diagonal matrices. Restricting to covariance matrices
that are diagonal is natural when modeling independent Gaussians.
We use the term {\em generic diagonal model} when $\mathcal{L}$ is
a generic point in the $(n-m)m$-dimensional Grassmannian
of $m$-dimensional subspaces inside the diagonal  $n \times n$ matrices.
\smallskip

\textbf{Brownian motion tree models:} A tree is a connected graph with no cycles. A rooted tree is obtained  by fixing a vertex, called the root, and directing all edges away from the root. Fix a rooted tree $T$ with $n$ leaves. Every vertex $v$ of $T$ defines a \emph{clade}, namely the set 
of leaves that are descendants of $v$. For the Brownian motion tree model on $T$,
   the space $\mathcal L$ is spanned by the rank-one matrices $e_Ae_A^T$, 
   where $e_A\in \{0,1\}^n$ is the indicator vector of $A$.
   Hence,
   if $\mathcal C$ is the set of all clades of $T$ then
\begin{equation}
\label{eq:treedef}
\Sigma\;=\;\sum_{A\in \mathcal C} \theta_A e_Ae_A^T,
\qquad \hbox{where $\theta_A$ are model parameters}.
\end{equation}
The linear equations for the subspace $\mathcal{L}$ are
$\sigma_{ij}=\sigma_{kl}$ whenever the least common ancestors ${\rm lca}(i,j)$ and ${\rm lca}(k,l)$ 
agree in the tree $T$. 
Assuming $\theta_A \geq 0$, the union of the models for all trees $T$ 
are characterized by the ultrametric condition 
  $\sigma_{ij}\geq \min\{\sigma_{ik},\sigma_{jk}\}\geq 0$. 
 Matrices of this form play an important role also in hierarchical clustering \cite[Section 14.3.12]{hastie2005elements}, phylogenetics \cite{felsenstein_maximum-likelihood_1973}, and random walks on graphs \cite{dellacherie}.

Maximum likelihood estimation for this class of models is generally complicated but recently there has been
progress (cf.~\cite{ane2014linear,sturmfels2019brownian}) on exploiting the nice structure of the matrices 
$\Sigma$ above.
 In Section~\ref{sec7} we study computational aspects of the MLE, and, more importantly, we provide a significant advance by considering the dual MLE.
 
 \smallskip
 
\textbf{Covariance graph models:}
We consider models $\mathcal{L}$ that arise from imposing zero restrictions on 
entries of $\Sigma$. This was studied in \cite{chaudhuri2007estimation, drton2002new}.
This is similar to Gaussian graphical models where zero restrictions are placed
on the inverse $\Sigma^{-1}$. 
We encode the sparsity structure with a graph whose edges correspond to nonzero off-diagonal entries of $\Sigma$. Zero entries in $\Sigma$ correspond to pairwise marginal independences. These arise
 in statistical modeling in the context of causal inference~\cite{cox1993linear}.
Models with zero restrictions on the covariance matrix are known as 
covariance graph models. Maximum likelihood in these Gaussian models can be carried out using Iterative Conditional Fitting \cite{chaudhuri2007estimation, drton2002new}, which is implemented in the \texttt{ggm} package in \textsc{R}~\cite{marchetti2006independencies}.

 \smallskip

\textbf{Toeplitz matrices:} Suppose $X=(X_1,\ldots,X_n)$ follows the autoregressive model of order 1, that is, $X_t=\rho X_{t-1}+\epsilon_t$, where $\rho\in \R$ and $\epsilon_t \sim N(0,\sigma)$ for some $\sigma$. Assume that the $\epsilon_t$ are mutually uncorrelated. Then ${\rm cov}(X_t,X_{t-k})=\rho^k$, and hence $\Sigma$ 
is a Toeplitz matrix. More generally, covariance matrices from stationary time series are Toeplitz. Multichannel and multidimensional  processes have covariance matrices of block Toeplitz form \cite{burg1982estimation,miller1987role}. Similarly, if $X$ follows the moving average 
process of order $q$ then ${\rm cov}(X_t,X_{t-k})=\gamma_k$ if $k\leq q$ and it is zero otherwise; see, for example, \cite[Section 3.3]{hamilton1994time}. Thus, in time series analysis, 
we encounter matrices like
\begin{equation}\label{eq:toeplitz}
\begin{small}
\begin{bmatrix}
	\gamma_0 & \gamma_1 & \gamma_2 & \gamma_3 & \gamma_4\\
	\gamma_1 & \gamma_0 & \gamma_1 & \gamma_2 & \gamma_3\\
	\gamma_2 & \gamma_1 & \gamma_0 & \gamma_1 & \gamma_2\\
	\gamma_3 & \gamma_2 & \gamma_1 & \gamma_0 & \gamma_1\\
	\gamma_4 & \gamma_3 & \gamma_2 & \gamma_1 & \gamma_0\\
\end{bmatrix}\qquad\mbox{ or }\qquad \begin{bmatrix}
	\gamma_0 & \gamma_1 & 0 & 0 & 0\\
	\gamma_1 & \gamma_0 & \gamma_1 & 0 & 0\\
	0 & \gamma_1 & \gamma_0 & \gamma_1 & 0\\
	0 & 0 & \gamma_1 & \gamma_0 & \gamma_1\\
	0 & 0 & 0 & \gamma_1 & \gamma_0\\
\end{bmatrix}	.
\end{small}
\end{equation}
We found that the ML degree for such models is surprisingly low.
This means that nonlinear algebra can reliably estimate Toeplitz matrices that are fairly large.
 
 \smallskip

\textbf{Colored covariance graph models:} A generalization of covariance graph models is
obtained by following  \cite{hojsgaard2008graphical}, 
which introduces graphical models with vertex and edge symmetries. Models of this type 
also generalize the Toeplitz matrices and the
 Brownian motion tree models. Following the standard convention we use the same colors for edges or vertices when the corresponding entries of $\Sigma$ are equal. The black color is considered neutral and encodes no restrictions. 

\begin{figure}[htp!]
	\includegraphics[scale=0.8]{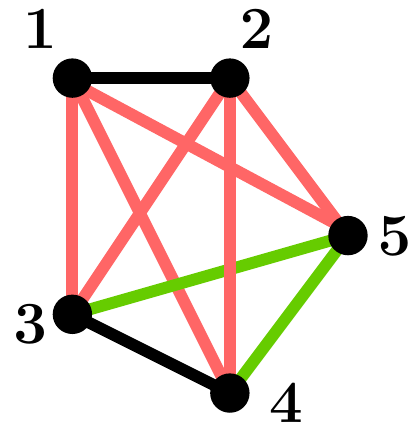}  \qquad
	\qquad\includegraphics[scale=0.5]{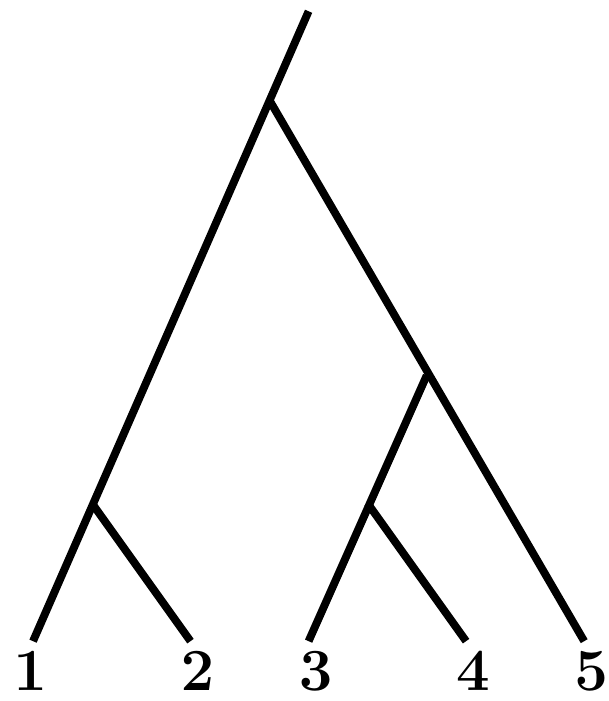}
	\caption{A covariance graph model with edge symmetries and the
	rooted tree for the  corresponding Brownian motion tree model.}\label{fig:treesymm}
\end{figure}

The Brownian motion tree model corresponds to a colored 
 model over the complete graph, where edge symmetries are encoded by the tree; cf.~Figure~\ref{fig:treesymm}. 
Also, both matrices in (\ref{eq:toeplitz}) 
represent covariance graph models with edge and vertex symmetries.

\section{Maximum likelihood estimator and its dual}\label{sec3}

Now that we have seen motivating examples, we formally define the 
MLE problem for a linear covariance model $\mathcal L$. 
Suppose we observe a random sample $\,X^{(1)},\ldots, X^{(N)}$ 
in $\mathbb{R}^n\,$ from $N_n(0,\Sigma)$. The sample covariance matrix is 
$S=\frac{1}{N}\sum_{i=1}^N X^{(i)}{X^{(i)}}^T$. 
The matrix $S$ is positive semidefinite. Our aim is to  maximize the function
$$ \qquad
\ell(\Sigma)\;=\;\log\det \Sigma^{-1}-{\rm tr}(S\Sigma^{-1})
\qquad \hbox{subject to $\Sigma\in \mathcal L$.}
$$
Following \cite[Proposition 7.1.10]{Sul}, this equals the
 log-likelihood function times $N/2$.

We fix the standard inner product $\langle A,B\rangle={\rm tr}(AB)$ on
the space $\mathbb S^n$ of symmetric matrices.
 The orthogonal complement $\mathcal L^\perp$ to a subspace $\mathcal L 
 \subset \mathbb S^n$ is defined as usual.
 
\begin{prop} \label{prop:primalMLE}
Finding all the critical points of the log-likelihood function
 amounts to solving the following system of linear and quadratic equations in $\,2 \cdot \binom{n+1}{2}$
 unknowns:
\begin{equation}\label{eq:opt_system}
 \Sigma\in \mathcal L, \qquad	K\Sigma=I_n,\qquad KSK-K\in \mathcal L^\perp.
\end{equation}
\end{prop}

\begin{proof}
The matrix $\Sigma$ is a critical point  $\ell$ if and only if, for every $U\in \mathcal L$, the 
derivative of $\ell$ at $\Sigma$ in the direction $U$ vanishes. This directional derivative equals
	$$
	-{\rm tr}(\Sigma^{-1}U)+{\rm tr}(S\Sigma^{-1}U\Sigma^{-1}).
	$$ 
	This formula follows by multivariate calculus from two facts: (i) the derivative of the matrix mapping $\Sigma\mapsto \Sigma^{-1}$ is the linear transformation $U\mapsto \Sigma^{-1}U\Sigma^{-1}$;
	 (ii) the derivative of the function $\Sigma\mapsto \log\det\Sigma$ is the linear functional $U\mapsto {\rm tr}(\Sigma^{-1}U)$. 
	 
	 Using the identity $K=\Sigma^{-1}$,  vanishing of the directional derivative is equivalent~to 
	$$
	-\langle K,U\rangle+\langle KSK,U\rangle\;=\;0.
	$$
 The condition $\langle KSK-K,U\rangle=0$ for all $U\in \mathcal L$ is
 equivalent to $KSK-K\in \mathcal L^\perp$. 
\end{proof}

\begin{exmp}[$3 \times 3$ Toeplitz matrices] \label{ex:toeplitz33primal}
Let $\mathcal{L}$ be the space of Toeplitz matrices
$$ \Sigma \,\, = \,\,
\begin{bmatrix}
	\gamma_0 & \gamma_1 & \gamma_2 \\
		\gamma_1 & \gamma_0 & \gamma_1 \\
			\gamma_2 & \gamma_1 & \gamma_0 
\end{bmatrix}.
$$
This space has dimension $3$ in $\mathbb{S}^3 \simeq \mathbb{R}^6$. 
Fix a sample covariance matrix $S = (s_{ij})$ with real entries.
We need to solve the  system (\ref{eq:opt_system}). This consists of
$3 + 9 + 3 = 15$ equations in $6+6=12$ unknowns, namely
the entries of the covariance matrix $\Sigma = (\sigma_{ij})$ and its
inverse $K = (k_{ij})$. 
The condition  $\, \Sigma\in \mathcal L\,$
gives three linear polynomials:
$$
\sigma_{11}- \sigma_{33} \,,\,\,\,  \sigma_{12}- \sigma_{23}\,,\,\,\,  \sigma_{22}- \sigma_{33}.
 $$
The condition 
 $K\Sigma=I_3$ translates into nine bilinear polynomials:
$$ \begin{small} \begin{matrix}
 \sigma_{11} k_{11}+ \sigma_{12} k_{12}+ \sigma_{13} k_{13}-1, \,
 \sigma_{12} k_{11} + \sigma_{22} k_{12}+ \sigma_{23} k_{13}, \,
      \sigma_{13} k_{11}+ \sigma_{23} k_{12}+ \sigma_{33} k_{13}, \\
      \sigma_{11} k_{12}+ \sigma_{12} k_{22}+ \sigma_{13} k_{23},\,
      \sigma_{12} k_{12}+ \sigma_{22} k_{22}+ \sigma_{23} k_{23}-1, \,
      \sigma_{13} k_{12}+ \sigma_{23} k_{22}+ \sigma_{33} k_{23}, \\
      \sigma_{11} k_{13}+ \sigma_{12} k_{23}+ \sigma_{13} k_{33}, \,
      \sigma_{12} k_{13}+ \sigma_{22} k_{23}+ \sigma_{23} k_{33},\,
      \sigma_{13} k_{13}+ \sigma_{23} k_{23}+ \sigma_{33} k_{33}-1.
\end{matrix}    \end{small}   $$
Finally,  the condition    
$ KSK-K\in \mathcal L^\perp$ translates into three quadratic polynomials:
$$ \begin{small}
\begin{matrix}
k_{11}^2 s_{11}+k_{12}^2 s_{11}+k_{13}^2 s_{11}+2 k_{11} k_{12} s_{12}
+2 k_{12} k_{22} s_{12}+2 k_{13} k_{23} s_{12}+2 k_{11} k_{13} s_{13} \\
+2 k_{12} k_{23} s_{13}  +2 k_{13} k_{33} s_{13}+k_{12}^2 s_{22}
+k_{22}^2 s_{22}+k_{23}^2 s_{22}+2 k_{12} k_{13} s_{23} 
+2 k_{22} k_{23} s_{23} \\ +2 k_{23} k_{33} s_{23}
+k_{13}^2 s_{33}+k_{23}^2 s_{33}+k_{33}^2 s_{33}-k_{11}-k_{22}-k_{33},
\end{matrix}
\end{small}
$$
$$
 \begin{small}
\begin{matrix}
      k_{23} s_{13} +k_{12} k_{33} s_{13}+k_{12} k_{22} s_{22}+k_{22} k_{23} s_{22}
      +k_{13} k_{22} s_{23}+k_{12} k_{23} s_{23} \\ +k_{23}^2 s_{23}
      +k_{22} k_{33} s_{23}+k_{13} k_{23} s_{33}+k_{23} k_{33} s_{33}-k_{12}-k_{23},
      \end{matrix}
\end{small}
$$

$$
 \begin{small}
\begin{matrix}
k_{11} k_{13} s_{11}+k_{12} k_{13} s_{12}+k_{11} k_{23} s_{12}+k_{13}^2 s_{13}
+k_{11} k_{33} s_{13} \\ 
+k_{12} k_{23} s_{22}+k_{13} k_{23} s_{23}+k_{12} k_{33} s_{23}
+k_{13} k_{33} s_{33}-k_{13}.
\end{matrix}
\end{small}
$$
The zero set of these $15$ polynomials in $12$ unknowns consists of three points
$(\hat\Sigma,\hat K)$. 
All three solutions are real and positive definite for the sample covariance matrix
\begin{equation}
\label{eq:saschamatrix} S \,\,  = \,\, \begin{small}
\begin{bmatrix}
  4/ 5 &  -9/ 5 &  -1/25 \\
 -9/5 &  79/16 & 25/24 \\
 -1/25 & 25/24 & 17/16
 \end{bmatrix}
 \,\, \approx \,\, \begin{bmatrix}
\phantom{-}0.8000 & -1.8000 & -0.0400 \\
 -1.8000 & \phantom{-}4.9375 & \phantom{-}1.0417\\
  -0.0400 & \phantom{-}1.0417 &   \phantom{-}  1.0625 
  \end{bmatrix}. \end{small}
 \end{equation}
 Namely, the three Toeplitz matrices  that solve the score equations
 for this $S$ are
 $$ \begin{matrix}
[\,\hat \gamma_0 \,,\,\, \hat \gamma_1\,,\,\, \hat \gamma_2 \,] && \quad \hbox{log-likelihood value} &  \\
 [2.52783, -0.215929, -1.45229] &&  -5.349366256855714 & \qquad \hbox{global maximum}  \\
  [2.39038, -0.286009, 0.949965] && -5.407087529857949 & \qquad \hbox{local maximum} \\
 [2.28596, -0.256394, 0.422321] &&  -5.415351037271449 & \quad \ \  \hbox{saddle point}
 \end{matrix}
 $$
So, even in this tiny example, our optimization problem has multiple local maxima
in the cone $\mathbb{S}^3_+$. A numerical
study of this phenomenon will be presented in Section~\ref{sec6}.
\end{exmp}     

Maximum likelihood is usually the preferred estimator 
in statistical analysis. However, for reasons like robustness and computational efficiency, it is customary 
 to also consider alternative estimators. A large family of estimators of general interest are 
 the M-estimators, which are all asymptotically normal and consistent under relatively minor conditions; see \cite[Sections 6.2-6.3]{huber2011robust}. A special case of an M-estimator is the 
 {\em dual maximum likelihood estimator} which we now discuss in more detail.

Dual estimation is based on the maximization of a dual likelihood function. In the Gaussian case this is motivated by interchanging the role of the parameter matrix $\Sigma$ and the empirical covariance matrix $S$. 
The {\em Kullback-Leibler divergence} of two Gaussian distributions $N(0,\Sigma_0)$ and 
$N(0,\Sigma_1)$ on $\mathbb{R}^n$ is equal to
 $$
 {\rm KL}(\Sigma_0,\Sigma_1)\;=\;\frac{1}{2}\left({\rm tr}(\Sigma_1^{-1}\Sigma_0)-n+\log\left(\frac{\det\Sigma_1}{\det\Sigma_0}\right)\right).
 $$
 Computing the MLE is equivalent to minimizing $ {\rm KL}(\Sigma_0,\Sigma_1)$ with respect to $\Sigma_1$ with $\Sigma_0=S$. On the other hand, the dual MLE is obtained by minimizing $ {\rm KL}(\Sigma_0,\Sigma_1)$ with respect to $\Sigma_0$ with $\Sigma_1=S$. Equivalently, we 
 set  $W=S^{-1}$ and we maximize
 $$
 \ell^\vee(\Sigma)\;=\;\log\det\Sigma-{\rm tr}(W\Sigma).
 $$

 The idea of utilizing the ``wrong'' Kullback-Leibler distance is ubiquitous in variational inference and is central for mean field approximation and related methods. The idea of using this estimation method for Gaussian linear  covariance models is very natural. It results in a unique maximum, since $\Sigma \mapsto \ell^\vee(\Sigma)$ is a convex 
 function on the positive definite cone $\mathbb{S}^n_+$.
 See \cite{christensen1989statistical} and also
  \cite[Section 3.2]{chaudhuri2007estimation} and
   \cite[Section 4]{kauermann1996dualization}.
   
 The following algebraic formulation is the analogue to Proposition~\ref{prop:primalMLE}.
 
\begin{prop} \label{prop:dualMLE}
Finding all the critical points of the dual log-likelihood function $\ell^\vee$
 amounts to solving the following system of equations in $\,2 \cdot \binom{n+1}{2}$
 unknowns:
\begin{equation}\label{eq:opt_system2}
 \Sigma\in \mathcal L, \qquad	K\Sigma=I_n,\qquad 
  K-W\in \mathcal L^\perp.
\end{equation}
\end{prop}

\begin{proof}
After switching the role of $K$ and $\Sigma$, and of $W$ and $S$, our problem becomes
 MLE for linear concentration models. Formula (\ref{eq:opt_system2}) is found in
 \cite[equation (10)]{SU}.
\end{proof}
 
 The next result lists properties of the dual MLE that are important for statistics.
 
 \begin{prop}\label{prop:dualstat}
 The dual maximum likelihood estimator of a Gaussian
 linear covariance model is consistent, asymptotically normal, and first-order efficient.	
 \end{prop}
 \begin{proof}
 See Theorem~3.1 and Theorem~3.2 in \cite{christensen1989statistical}.
 \end{proof}

 First-order efficiency means that the asymptotic variance of the properly normalized dual MLE is optimal, that is, it equals
  the asymptotic variance of the MLE. For finite samples the dual MLE is less efficient than the MLE, 
  but the computational cost of the MLE is higher.
   In practice, one rarely
  has access to the real MLE and one uses a local maximum of the likelihood function.
  Thus, the advantage in efficiency is smaller than predicted by the theory. 
 The dual MLE offers a convenient alternative.
 
In this paper, we focus on algebraic structures, and we note the following
important distinction between our two estimators. The
MLE requires the quadratic equations $KSK-K\in \mathcal L^\perp$ in (\ref{eq:opt_system}),
whereas the dual MLE requires the linear equations  $ K-W\in \mathcal L^\perp$ in (\ref{eq:opt_system2}).
The latter are easier to solve than the former, and they give far fewer solutions.
This is quantified by the tables for the ML degrees in the next sections.

We are particularly interested in models whose dual ML estimator $(\check \Sigma, \check K)$ 
can be written as
 an explicit expression in the sample covariance matrix $S$.  We identify such scenarios
in Sections \ref{sec6} and \ref{sec7}. Here is a first example to illustrate this point.

\begin{exmp}
We revisit the Toeplitz model in Example~\ref{ex:toeplitz33primal}.
For the dual MLE, the three quadratic polynomials in $K$ are now replaced by
three linear polynomials:
$$
 k_{11}+k_{22}+k_{33}-w_{11}-w_{22}-w_{33} \, ,\,\,
k_{12}+k_{23}-w_{12}-w_{23} \, , \,\,
k_{13}-w_{13}.
$$
The $w_{ij}$ are the entries of the inverse sample covariance matrix $W = S^{-1}$.
The new system has two solutions, and we can write the $\check \sigma_{ij}$ and
$\check k_{ij}$ in terms of the $w_{ij}$ (or the $s_{ij}$) using the familiar formula
for solving quadratic equations in one variable.
Specifically, for the covariance matrix $S$ in
(\ref{eq:saschamatrix}) we find that the dual MLE is given~by
$$ \begin{matrix}
[\,\check \gamma_0, \,\check \gamma_1,\,\check \gamma_2\,] \,\, = \,\,
 \bigl[  0.203557267562 \,,\,\,
   -0.189349961613\,,\,\,
   0.1963649733282  \, \bigr] \qquad \qquad \quad  \medskip  \\ = \,\,
\bigl[ \,{\frac{1284368265268038839512}{12363704694314904961417}}+{\frac {52\,\sqrt {561647777654592987689702150027364667081}}{
12363704694314904961417}} \, , \qquad \medskip \\ \qquad
-{\frac{5817390611804320873051}{61818523471574524807085}}-{\frac {655679934637\,\sqrt {
561647777654592987689702150027364667081}}{163146905524715599705244729886305}} \, ,
\medskip
\\
\qquad \qquad \,\,
{\frac{1990451408446510673691859}{
22254668449766828930550600}}+{\frac {264990063915733\,\sqrt {561647777654592987689702150027364667081}}{
58732885988897615893888102759069800}} \, \bigr].
\end{matrix}
$$
Needless to say, nonlinear algebra goes much beyond the quadratic formula.
In what follows we shall employ state-of-the-art methods for solving polynomial equations.
\end{exmp}

\section{General Linear Constraints}\label{sec4}

The \emph{maximum likelihood degree} 
of a linear covariance model 
$\mathcal{L}$ is, by definition, the number of complex solutions
to the likelihood equations \eqref{eq:opt_system} for generic data $S$.
This is abbreviated {\em ML degree}; see \cite[Section 7.1]{Sul}.
To compute the ML degree, take $S$ to be a random symmetric $n \times n$ matrix and
 count all complex critical points of the likelihood function $\ell(\Sigma)$ for $\Sigma \in \mathcal{L}$.
Equivalently, the ML degree of the model $\mathcal{L}$ is the number of complex
solutions $(\Sigma, K)$ to the polynomial equations in~(\ref{eq:opt_system}).                                           

We also consider the complex critical points of
 the dual likelihood function $\ell^\vee(\Sigma)$.
Their number, for a generic matrix $S \in \S^n$, is the {\em dual ML degree} of $\mathcal{L}$.
It coincides with the number of  complex
solutions $(\Sigma, K)$ to the polynomial equations in~(\ref{eq:opt_system2}).

Our ML degrees can be computed symbolically in a computer algebra
system that rests on Gr\"obner bases.  However, this approach is limited to small instances.
To get further, we use the methods
from numerical nonlinear algebra described in Section~\ref{sec5}.

We here focus on a generic $m$-dimensional linear
subspace $\mathcal{L}$  of $\S^n$. In practice this means that 
a basis for $\mathcal{L}$ is chosen by sampling $m$ matrices at random from $\S^n$.

\begin{prop}
The ML degree and the dual ML degree of a generic subspace $\mathcal{L}$
of dimension $m$ in $\S^n$
depends only on $m$ and $n$. It is independent
of the particular choice of $\mathcal{L}$.
For small parameter values, these ML degrees 
are listed in Table \ref{table:ml-generic}.
\end{prop}

\begin{proof} The independence rests on general results in
algebraic geometry, to the effect that the computation can be
done over the rational function field where the coordinates of $\mathcal{L}$
and $S$ are unknowns. The ML degree will be the same for all specializations
to $\mathbb{R}$ that remain outside a certain discriminant hypersurface.
Table~\ref{table:ml-generic} and further values are computed 
rapidly using the software described in Section \ref{sec5}.
\end{proof}

The dual ML degree was already studied by
 Sturmfels and Uhler in \cite[Section 2]{SU}.
Our table on the right  is in fact found in their paper.
The symmetry along its columns is proved in \cite[Theorem 2.3]{SU}.
It states that the dual ML degree for dimension $m$ coincides with 
the dual ML degree for codimension $m-1$. This is derived from the
equations  (\ref{eq:opt_system2})  by an appropriate homogenization.
Namely,  the middle equation is clearly symmetric under switching
the role of $K$ and $\Sigma$, and
the linear equations on the left and on the right
in (\ref{eq:opt_system2}) can also be interchanged under this switch.

\begin{table}[]
\begin{tabular}{@{}c| ccccc  @{}}
	\toprule
	{\multirow{2}{*}{m}} & \multicolumn{5}{c}{n} \\
	  & 2 & 3 & 4 & 5 & 6 \\ \midrule
	2 & 1 & 3 & 5 & 7 & 9 \\
	3 & 1 & 7 & 19 & 37 & 61  \\
	4 &   & 7 & 45 & 135 & 299 \\
	5 &   & 3 & 71 & 361 & 1121 \\
	6 &   & 1 & 81 & 753 & 3395 \\
	7 &   &   & 63 & 1245 & 8513 \\
	8 &   &   & 29 & 1625 & 17867 \\
	9 &   &   & 7 & 1661 & 31601 \\
	10 &   &   & 1 & 1323 & 47343 \\
	11 &   &   &  & 801 & 60177 \\
	12 &   &   &  & 347 & 64731 \\
	13 &   &   &  & 97 & 58561 \\
	14 &   &   &  & 15 & 44131 \\
	15 &   &   &  & 1 & 27329 \\
	16 &   &   &  &   & 13627 \\
	17 &   &   &  &   & 5341 \\
	18 &   &   &  &   & 1511 \\
	19 &   &   &  &   & 289 \\
	20 &   &   &  &   & 31 \\
	21 &   &   &  &   & 1 \\
	\bottomrule
\end{tabular} \qquad  \qquad \qquad\begin{tabular}{@{}c| ccccc @{}}
	\toprule
	{\multirow{2}{*}{m}} & \multicolumn{5}{c}{n} \\
	  & 2 & 3 & 4 & 5 & 6 \\ \midrule
	2 & 1 & 2 & 3 & 4 & 5 \\
	3 & 1 & 4 & 9 & 16 & 25 \\
	4 &   & 4 & 17 & 44 & 90 \\
	5 &   & 2 & 21 & 86 & 240 \\
	6 &   & 1 & 21 & 137 & 528 \\
	7 &   &   & 17 & 188 & 1016 \\
	8 &   &   & 9 & 212 & 1696  \\
	9 &   &   & 3 & 188 & 2396 \\
	10 &   &   & 1 & 137 & 2886 \\
	11 &   &   &  & 86 & 3054 \\
	12 &   &   &  & 44 & 2886 \\
	13 &   &   &  & 16 & 2396 \\
	14 &   &   &  & 4 & 1696 \\
	15 &   &   &  & 1 &  1016 \\
	16 &   &   &  &   & 528 \\
	17 &   &   &  &   & 240 \\
	18 &   &   &  &   & 90 \\
	19 &   &   &  &   & 25 \\
	20 &   &   &  &   & 5 \\
	21 &   &   &  &   & 1 \\
	\bottomrule
\end{tabular} 
\medskip
\caption{ML degrees and dual ML degrees for generic models}\label{table:ml-generic}
\end{table}

It was conjectured in \cite[Section 2]{SU} that, for fixed $m$, the dual ML degree is
a polynomial of degree $m-1$ in the matrix size $n$. This is easy to see for $m \leq 3$.
The polynomials for $m=4$ and $m=5$ were also derived in \cite[Section 2]{SU}.

The situation is similar but more complicated for the ML degree. First of all,
the symmetry along columns no longer holds as seen on the left in Table~\ref{table:ml-generic}.
This is explained by the fact that the linear equation $K - W \in \mathcal{L}^\perp$ is now
replaced by the quadratic equation $KSK-K \in \mathcal{L}^\perp$. However,
the polynomiality along the rows of Table \ref{table:ml-generic} seems to persist.
For $m=2$ the ML degree equals $2n-3$, as shown recently by
Coons, Marigliano and Ruddy~\cite{CMR}.
For $m \geq 3$ we propose the following conjecture.

\begin{conj} The ML degree of a linear covariance model of dimension $m$
is a polynomial of degree $m-1$ in the ambient dimension $n$.
For $m=3$ this ML degree equals $\,3n^2-9n+7$,
and for $m=4$ it equals $\,11/3n^3-18n^2+85/3n-15$.
\end{conj}

We now come to {\em diagonal linear covariance models}. For these models,
 $\mathcal{L}$ is a linear subspace of dimension $m$ inside
the space $\R^n$ of diagonal $n \times n$-matrices.
We wish to determine the ML degree and dual ML degree
when $\mathcal{L}$ is generic in $\R^n$.

In the diagonal case, the score equations simplify as follows.
Both the covariance matrix and the concentration matrix are diagonal.
We eliminate the entries of $\Sigma$ by setting
$K = {\rm diag}(k_1,\ldots,k_n)$ and
$\Sigma = {\rm diag}(k_1^{-1},\ldots,k_n^{-1})$.
We also write $s_1,\ldots,s_n$ for the diagonal
entries of the sample covariance matrix $S$, and
$w_i = s_i^{-1}$ for their reciprocals.
Finally, let $\mathcal{L}^{-1}$ denote the {\em reciprocal linear space} of $\mathcal{L}$,
i.e.~the variety obtained as  the closure of the set of
   coordinatewise reciprocals of vectors in $\mathcal{L} \cap (\mathbb{R}^*)^n$.

\begin{table}[h]
\begin{tabular}{@{}c| ccccc  @{}}
	\toprule
	{\multirow{2}{*}{m}} & \multicolumn{5}{c}{n} \\
	  & 3 & 4 & 5 & 6 & 7 \\ \midrule
	2 & 3 & 5 & 7 & 9 & 11 \\
	3 & 1 & 7 & 17 & 31 & 49  \\
	4 &   & 1 &  15 & 49 & 111 \\
	5 &   &    &    1 &  31 & 129 \\
	6 &   &   &       &      1 & 63 \\
	7 &   &   &       &         & 1 \\
		\bottomrule
\end{tabular} \qquad \qquad \qquad\begin{tabular}{@{}c| ccccc @{}}
	\toprule
	{\multirow{2}{*}{m}} & \multicolumn{5}{c}{n} \\
	  & 3 & 4 & 5 & 6 & 7 \\ \midrule
	2 & 2 & 3 & 4 & 5 & 6 \\
	3 & 1 & 3 & 6 & 10 & 15 \\
	4 &   & 1 &  4  & 10 & 21 \\
	5 &   &    &  1 & 5 & 21\\	
	6 &   &    &     &  1 & 15 \\
	7 &   &   &       &         & 1 \\
	\bottomrule
\end{tabular} 
\medskip
\caption{ML degrees and dual ML degrees for generic diagonal models}\label{table:ml-generic-diag}
\end{table}

\begin{prop}
Let $\mathcal{L} \subset \mathbb{R}^n$ be a linear space,
viewed as a Gaussian covariance model of diagonal matrices.
The score equations for the likelihood in (\ref{eq:opt_system})  
and the dual likelihood  in (\ref{eq:opt_system2})  can be written as systems of
$n$ equations in $n$ unknowns as follows:
\begin{itemize}
\item[(\ref{eq:opt_system}')] \qquad
$ (k_1,\ldots,k_n) \in \mathcal{L}^{-1}$ \ and  \
$(
s_1 k_1^2-k_1,
s_2 k_2^2-k_2,
\ldots,
s_n k_n^2-k_n ) \in \mathcal{L}^\perp, $
\item[(\ref{eq:opt_system2}')] \qquad
$(k_1,\ldots,k_n) \in \mathcal{L}^{-1}$ \  and  \
$(k_1-w_1,k_2-w_2,\ldots,k_n-w_n) \in \mathcal{L}^\perp $.
\end{itemize}
The number of complex solutions to  (\ref{eq:opt_system2}')
for generic $\mathcal{L}$ of dimension $m$ equals
$\binom{n-1}{m-1}$.
\end{prop}

\begin{proof}
The translation of 
(\ref{eq:opt_system}) and (\ref{eq:opt_system2})
to (\ref{eq:opt_system}') and (\ref{eq:opt_system2}') is straightforward.
The equations (\ref{eq:opt_system2}') represent a general linear section
of the reciprocal linear space $\mathcal{L}^{-1}$.
 Proudfoot and Speyer showed that the degree
of $\mathcal{L}^{-1}$ equals the M\"obius invariant of the underlying matroid.
We refer to \cite{KV} for a recent study and many references. This M\"obius invariant
equals $\binom{n-1}{m-1}$ in the generic case, when the matroid is uniform.
\end{proof}

It would be desirable to express
the number of complex solutions to  (\ref{eq:opt_system}')
as a matroid invariant, and thereby explain the
entries on the left side of Table \ref{table:ml-generic-diag}.
As before, the $m$th row gives the values of a polynomial
of degree $m-1$. For instance, for $m=3$ we find
$\,2n^2 - 8n + 7$, and for $m=4$ we find $\, 4/3 n^3 - 10 n^2 + 68/3 n - 15$.

\section{Numerical Nonlinear Algebra}\label{sec5}

Linear algebra is the foundation of  scientific computing and applied mathematics.
\emph{Nonlinear algebra} \cite{nonlinear-algebra} is a generalization where linear systems are replaced
by nonlinear equations and inequalities.
At the heart of this lies algebraic geometry, but there are links to many other branches, such as combinatorics, algebraic topology, commutative algebra, convex and discrete geometry, tensors and multilinear algebra, number theory and representation theory.
Nonlinear algebra is not simply a rebranding of algebraic geometry.
It highlights that the focus is on computation and applications, and the theoretical needs that this requires, results in a new perspective.

We refer to \emph{numerical nonlinear algebra} as the branch of nonlinear algebra which is concerned with the efficient numerical solution of polynomial equations and inequalities. In the existing literature, this is 
referred to as numerical algebraic geometry. In the following we discuss the numerical solution of
polynomial equations, and we describe the techniques used for deriving the computational results in this paper. 

One of our main contributions is the
{\tt Julia} package {\tt LinearCovarianceModels.jl}
for estimating linear covariance models; see (\ref{eq:uurrll}).
Given $\mathcal{L}$, our package computes the ML degree,
and the dual ML degree. For any $S$, it finds
all critical points and it selects those that are local maxima. 
The following example explains how this is done.

 \begin{exmp}
  We use the package to verify Example \ref{ex:toeplitz33primal}:\\[2mm]
\indent\texttt{\julia  using LinearCovarianceModels} \\
\indent\texttt{\julia $\Sigma$ $=$ toeplitz(3)}\\
\indent\texttt{3-dimensional LCModel:} \\
\indent\;$\theta_1\; \theta_2\; \theta_3$ \\
\indent\;$\theta_2\; \theta_1\; \theta_2$ \\
\indent\;$\theta_3\; \theta_2\; \theta_1$ \\[2mm]
\noindent We compute the ML degree of the family $\Sigma$ by computing all solutions for a generic instance.
The pair of solutions and generic instance is called an {\em ML degree witness}.\\[2mm]
\indent\texttt{\julia W = ml\_degree\_witness($\Sigma$)} \\
\indent\texttt{MLDegreeWitness:} \\
\indent\texttt{\;$\circ$ ML degree $\rightarrow$ 3} \\
\indent\texttt{\;$\circ$ model dimension $\rightarrow$ 3} \\
\indent\texttt{\;$\circ$ dual $\rightarrow$ false} \\[2mm]
\noindent By default, the computation of the ML degree witness relies on a heuristic stopping criterion.
We can numerically verify the correctness by using a trace test \cite{leykin2018trace}.\\[2mm]
\indent\texttt{\julia verify(W)} \\
\indent\texttt{Compute additional witnesses for completeness...} \\
\indent\texttt{Found 10 additional witnesses} \\
\indent\texttt{Compute trace...} \\
\indent\texttt{Norm of trace: 2.6521474798326718e-12} \\
\indent\texttt{\bf true} \\[2mm]
\noindent
We now input the specific sample covariance matrix in \eqref{eq:saschamatrix}, and
we compute all critical points of this MLE problem using the ML degree witness 
from the previous step. \\[2mm]
\indent\texttt{\julia S = [4/5 -9/5 -1/25; -9/5 79/16 25/24; -1/25 25/24 17/16]}; \\
\indent\texttt{\julia critical\_points(W, S)} \\
\indent\texttt{3-element Array\{Tuple\{Array\{Float64,1\},Float64,Symbol\},1\}:} \\
\indent\texttt{\;([2.39038, -0.286009, 0.949965], -5.421751313919751, :local\_maximum)} \\ 
\indent\texttt{\;([2.52783, -0.215929, -1.45229], -5.346601549034418, :global\_maximum)} \\
\indent\texttt{\;([2.28596, -0.256394, 0.422321], -5.424161999175718, :saddle\_point)} \\[2mm]
\noindent If only the global maximum is of interest
then this can also be computed directly.\\[2mm]
\indent\texttt{\julia mle(W, S)} \\
\indent\texttt{3-element Array\{Float64,1\}:} \\
\indent\texttt{\;2.527832268219689 } \\
\indent\texttt{\;-0.21592947057775033} \\
\indent\texttt{\;-1.4522862659134732}
\end{exmp}
\noindent By default only positive definite solutions are reported. To list all critical points we run the command with an additional option. \\[2mm]
\indent\texttt{\julia critical\_points(W, S, only\_positive\_definite=false)} \\[2mm]
In this case, since the ML degree is 3, we are not getting more solutions. 
\bigskip

In the rest of this section we explain the mathematics behind our software,
and how it applies to our MLE problems. A textbook introduction 
to the~numerical solution of polynomial systems by homotopy continuation methods
is \cite{SW}.

Suppose we are given $m$
 polynomials $f_1,\ldots,f_m$ in $n$ unknowns $x_1,\ldots,x_n$ 
 with complex coefficients, where $m \ge n$. We are interested in computing all isolated 
 complex solutions of the system $f_1(x)=\cdots=f_m(x)=0$.
These solutions comprise the zero-dimensional components of the variety $V(F)$ 
where $F=(f_1,\ldots,f_m)$. 

The general idea of homotopy continuation is as follows.
Assume we have another system $G=(g_1,\ldots,g_m)$ of polynomials for which we know some or all of its solutions.
Suppose there is a \emph{homotopy} $H(x,t)$ with $H(x,0) = G(x)$, $H(x,1)=F(x)$ 
with the property that,
for every $x^* \in V(G)$, there exists a smooth path $x: [0,1) \rightarrow \mathbb{C}^n$ with $x(0) = x^*$
and $H(x(t),t)=0$ for all $t \in [0,1)$.
Then we can track each point in $ V(G)$
to a point in  $V(F)$. This is done by solving the {\em Davidenko differential equation}
$$\frac{\partial H}{\partial x}\bigl(x(t),t\bigr) \cdot \dot{x}(t) \,+\, \frac{\partial H}{\partial t}\bigl(x(t),t \bigr) \,\,= \,\, 0 $$
with initial condition $x(0) = x^*$. Using a \emph{predictor-corrector} scheme for  numerical path tracking, both the local and global error can be controlled.
Methods called \emph{endgames} are used 
to handle divergent paths and singular solutions \cite[Chapter~10]{SW}.

Here is a general framework for start systems and homotopies. 
Embed $F$ in a family of polynomial systems $\mathcal{F}_Q$, continuously
parameterized by a convex open set $Q \subset \mathbb{C}^k$. 
We have $F=F_q \in \mathcal{F}_Q$ for some $q \in Q$.
Outside a Zariski closed set $\Delta \subset Q$, every system in $\mathcal{F}_Q$ 
has the same number of solutions. If  $p\in Q \backslash \Delta$ then
 $F_p$ is such a {\em generic instance} of the family $\mathcal{F}_Q$,
and the following is a suitable homotopy \cite{Morgan:Sommese:87}:
\begin{equation}\label{eq:parameter_homotopy}
	H(x,t) \,=\, F_{(1-t)p+tq}(x) \,.
\end{equation}

Now, to compute $V(F_q)$, it suffices to find all solutions of a generic instance $F_p$ 
and then track these along the homotopy  \eqref{eq:parameter_homotopy}.
Obtaining all solutions of a generic instance can be a challenge, but this has to be done \emph{only once}!
That is the \emph{offline phase}.
 Tracking from a generic to a specific instance of interest is the \emph{online phase}.

A key point in applying this method is the choice of the family $\mathcal{F}_Q$.
For MLE problems in statistics, it is natural to choose $Q$
as the space of data or instances. In our scenario,
$Q$ is $\mathbb{S}^n$, or a complex version thereof.
  We shall discuss this below.

First, we explain the {\em monodromy method} for an 
arbitrary family $\mathcal{F}_Q$.
Suppose the general instance has $d$  solutions, and that we are given one
 \emph{start pair} $(x_0, p_0)$. This means that $x_0$ is a solution to the instance $F_{p_0}$.
 Consider the incidence variety
$$Y \,:=\, \bigl\{(x,p) \in \mathbb{C}^n \times Q \; | \; F_p(x) = 0\bigr\}. $$
Let $\pi$ be the projection from $\mathbb C^n\times Q$ onto the second factor. 
For $q \in Q\backslash \Delta$, the fiber
$\pi^{-1}(q)$ has exactly $d$ points.
A loop in $Q \backslash \Delta$ based at $q$ has $d$ lifts to $Y$.
 Associating a point in the fiber to the endpoint of the corresponding lift gives a permutation in $S_d$. This defines an action of the fundamental group of $Q\backslash \Delta$ on the fiber $\pi^{-1}(q)$. The \emph{monodromy group} of our family is the image of the fundamental group in $S_d$.
 
The monodromy method fills the fiber $\pi^{-1}(p_0)$ by exploiting the monodromy group. For this,
 the start solution $x_0$ is numerically tracked along a loop in $Q\backslash \Delta$,
  yielding a solution $x_1$ at the end. If $x_1\not= x_0$, then $x_1$ is also tracked along the \emph{same} loop,
   possibly yielding again a new solution. This is done until no more solutions are found. 
   Then, all solutions are tracked along a new loop, where the process is repeated. This process 
is stopped by use of a \emph{trace test}.
For a detailed description of the monodromy method and the trace test 
see \cite{delCampo:Rodriguez:2017,leykin2018trace}.
To get this off the ground, one needs a start pair $(x_0,p_0)$. This can often be found by
inverting the problem. Instead of finding a solution $x_0$ to a given $p_0$, we start with $x_0$ and look for $p_0$ such that $F_{p_0}(x_0)=0$.

We now explain how this works for the score equations \eqref{eq:opt_system} of our MLE problem.
First pick a random matrix $\Sigma_0$ in the subspace $\mathcal{L}$.
We next compute $K_0$ by inverting $\Sigma_0$.
Finally we need to find a symmetric matrix $S_0$ such that
$\,K_0S_0K_0-K_0\in \mathcal L^\perp$. Note that this is a linear system
of equations and hence directly solvable. In this manner, we easily find
a start pair $(x_0,p_0)$ by setting $p_0 = S_0$  and $x_0 = (\Sigma_0,K_0)$.

The number $d$ of solutions to a generic instance is the ML degree of our model.
A priori knowledge of $d$ is useful because it serves as a stopping
criterion in the monodromy method. This is one reason for 
focusing on the ML degree in this paper.

\section{Local Maxima versus Rational MLE}\label{sec6}

The theme of this paper is maximum likelihood inference for  linear covariance models.
We developed some numerical nonlinear algebra for this problem, and
we offer a software package (\ref{eq:uurrll}). From the applications perspective,
this is motivated by the fact that the likelihood function is non-convex. It can have
multiple local maxima. A concrete instance for $3 \times 3$ Toeplitz matrices was shown in
Example \ref{ex:toeplitz33primal}.

In this section we undertake a more systematic experimental study of local maxima.
Our aim is to answer the following question:
there is the theoretical possibility that
$\ell(\Sigma)$ has  many local maxima, but can we also observe this in practice?

To address this question, we explored a range of linear covariance models $\mathcal{L}$.
For each model, we conducted the following experiment. We repeatedly
generated sample covariance matrices $S \in \mathbb S^n_+$.
This was done as follows. We first sample
 a matrix $X \in \R^{n\times n}$ by picking each entry independently from a normal distribution 
 with mean zero and variance one. And, then we set $S:=XX^T/n$. 
  This is equivalent~to sampling $nS\in \mathbb S^n_+$ from the standard Wishart distribution
  with $n$ degrees of freedom.
 
\begin{table}[h]
	\small
	\begin{tabular}{@{}l| ccccc ccccc ccc @{}}
		\toprule
		   & \multicolumn{13}{c}{m} \\
		  & 2 & 3 & 4 & 5 & 6 & 7 & 8 & 9 & 10 & 11 & 12 & 13 & 14 \\ \midrule
		ML degree    & 7 & 37 & 135 & 361 & 753 & 1245 & 1625 & 1661 & 1323 & 801 & 347 & 97 & 15 \\ \midrule
		max          & 2 & 3 & 3 & 5 & 5 & 5 & 5 & 6 & 7 & 5 & 4 & 2 & 1 \\
		max pd      & 1 & 2 & 3 & 3 & 4 & 4 & 4 & 4 & 5 & 5 & 4 & 2 & 1 \\ \midrule
		multiple     & 0.4\% & 5.8 & 13.8 & 31.2 & 37.2 & 39.0 & 40.6 & 37.4 & 32.0 & 20.4 & 13.8 & 3.0 & 0.0 \\
		multiple pd & 0.0\% & 4.6 & 11.2 & 22.4 & 25.2 & 31.6 & 33.0 & 34.8 & 29.6 & 19.4 & 13.0 & 3.0 & 0.0 \\
			\bottomrule   
			\smallskip
	\end{tabular}  
\caption{
Experiments for generic $m$-dimensional linear subspaces of $\mathbb S^5$. }\label{tab:exp1}
	\end{table}

For each of the generated sample covariance matrices $S$, we computed 
the real solutions of the likelihood
equations (\ref{eq:opt_system}). From these, we identified the set of
all local maxima in $\mathbb S^n$,
and we extracted its subset of local maxima in the positive definite cone
${\mathbb S}_+^n$. We recorded the numbers of these local maxima.
 Moreover, we kept track of the fraction of instances $S$ for which there were
    multiple (positive define) local maxima.
In Table \ref{tab:exp1} we present our results for $n=5$ and generic linear subspaces
$\mathcal{L}$.

For each $m$ between $2$ and $14$,
we selected five generic linear subspaces $\mathcal{L}$ in the $15$-dimensional space $\mathbb S^5$.
Each linear subspace $\mathcal{L}$ was constructed by choosing a basis of positive definite matrices.
   The basis elements were constructed with the same sampling method as the sample covariance matrices.
The ML degree of this linear covariance model is the corresponding entry in the $n=5$ column on 
the left in Table \ref{table:ml-generic}.
These degrees are repeated in the 
row named {\em ML degree} in Table \ref{tab:exp1}.

For each model $\mathcal{L}$, we generated $100$ sample covariance matrices $S$,
and we solved the likelihood equations (\ref{eq:opt_system}) 
using our software \ {\tt LinearCovarianceModels.jl}.
The row \emph{max} denotes the largest number of local maxima 
that was observed in these $100$ experiments. The row
\emph{multiple} gives the fraction of instances which resulted in 
two or more local maxima. These two numbers pertain to local maxima in $\mathbb{S}^5$.
The rows \emph{max pd} and \emph{multiple pd} are the analogues 
restricted to the positive definite cone $\mathbb{S}^5_+$.

For an illustration, let us discuss the models of dimension $m=7$. These equations
(\ref{eq:opt_system}) 
have $1245$ complex solutions, but the number of real solutions is much smaller.
Nevertheless, in two fifths of the instances ($39.0 \%$) there were two or more
local maxima in $\mathbb{S}^5$. In one third of the instances  ($ 31.6 \%$)
the same happened $\mathbb{S}^5_+$. The latter is
  the case of interest in statistics.
  One instance had four local maxima in $\mathbb{S}^5_+$.

\smallskip

The second experiment we report concerns a combinatorially defined class
of linear covariance models, namely the Brownian motion tree models  in (\ref{eq:treedef}).
We consider eleven combinatorial types of trees with $5$ leaves. For each model we perform the
 experiment described above, but we now
used $500$ sample covariance matrices per model. Our results are presented in Table \ref{tab:ex2}, in the same format
as in Table~\ref{tab:exp1}.

\begin{table}[h]
	\begin{tabular}{@{}l| ccccc ccccc c @{}}
		\toprule
		   & \multicolumn{11}{c}{tree number} \\
		  & 1 & 2 & 3 & 4 & 5 & 6 & 7 & 8 & 9 & 10 & 11  \\ \midrule
		ML degree    & 37 & 37 & 81 & 31 & 27 & 31 & 31 & 27 & 13 & 17 & 17 \\ \midrule
		max          & 3 & 3 & 4 & 3 & 3 & 3 & 4 & 3 & 3 & 3 & 3 \\
		max pd      & 3 & 2 & 3 & 3 & 3 & 3 & 2 & 2 & 3 & 3 & 3 \\ \midrule
		multiple     & 21.2\% & 22.8 & 24.2 & 15.6 & 23.0 & 21.2 & 21.2 & 15.4 & 13.8 & 16.2 & 12.4 \\
	   multiple pd   & 8.2\%  & 9.4  & 14.0 & 10.0 & 15.8 & 13.0 & 12.2 & 8.8  & 13.8 & 16.2 & 12.4 \\
			\bottomrule
	\end{tabular} 
	\medskip
	\caption{Experiments for eleven Brownian motion tree models with $5$ leaves.	} \label{tab:ex2}
\end{table}

The eleven trees are numbered by the order in which they appear in Table \ref{tab:BMT}.
For instance, tree 1 gives the $7$-dimensional model in $\mathbb{S}^5_+$ 
whose covariance matrices are

$$ \begin{small}
\Sigma \,\,\,=\,\,\,
\begin{bmatrix} 
\gamma_1 & \gamma_6 & \gamma_6 & \gamma_6 & \gamma_7 \\
\gamma_6 & \gamma_2 & \gamma_6 & \gamma_6 & \gamma_7 \\
\gamma_6 & \gamma_6 & \gamma_3 & \gamma_6 & \gamma_7 \\
\gamma_6 & \gamma_6 & \gamma_6 & \gamma_4 & \gamma_7 \\
\gamma_7 & \gamma_7 & \gamma_7 & \gamma_7 & \gamma_5 \\
\end{bmatrix}.
\end{small}
$$
This model has ML degree $37$. Around eight percent of the instances
led to multiple maxima among positive definite matrices.
Up to three such maxima were observed. 


\smallskip

The results reported in  Tables~\ref{tab:exp1} and  \ref{tab:ex2}
show that the maximal number of local maxima increases with the ML degree.
But, they do not increase as fast as one would expect from the growth of the ML degree. On the other hand, the frequency of observing multiple local maxima seems to be closer correlated to the ML degree.

Here is an interesting observation to be made in Table~\ref{tab:ex2}.
The last three trees, labeled 9, 10 and 11, are the binary trees.
These have the maximum dimension  $2n-2$.
For these models, every local maximum in
$\mathbb{S}^n$ is also in the positive definite cone $\mathbb{S}^n_+$.
We also verified this for all binary trees with $n=6$ leaves. This is interesting since the positive definiteness constraint is the hardest to respect in an optimization routine.
It is tempting to conjecture that this persists for all binary trees with $n \geq 7$.

There is another striking observation in
Table \ref{tab:BMT}. The dual ML degree for binary trees is always equal to one.
We shall prove in Theorem \ref{thm:bmtdual}
  that this holds for any~$n$.
  This means that the dual MLE can be expressed as a rational function
in the data $S$. Hence there is only one local maximum,
which is therefore the global maximum.

\smallskip

We close this section with a few remarks on the important 
special case when the ML degree or the dual ML degree is equal to one.
This holds if and only if each entry of the estimated matrix
$\hat \Sigma$ or $\check \Sigma$ is a rational function in the 
$\binom{n+1}{2}$ quantities $s_{ij}$.

Rationality of the MLE has received a lot of attention in the case of {\em discrete random variables}.
See \cite[Section 7.1]{Sul} for a textbook reference.
If the MLE of a discrete model is rational then its coordinates
are alternating products of linear forms in the data \cite[Theorem 7.3.4]{Sul}.
This result due to Huh was refined in \cite[Theorem 1]{DMS}.
At present we have no idea what the analogue in the Gaussian case might look like.

\begin{prob} \label{prob:OrlandosProblem}
Characterize all Gaussian models whose MLE is a rational function.
\end{prob}

In addition to the binary trees in Theorem \ref{thm:bmtdual},
statisticians are familiar with a number of 
situations when the dual MLE is rational.
The dual MLE is the MLE of a linear concentration model with the sample covariance matrix 
$S$ replaced by its inverse $W$. This is  studied in \cite{SU}
and in many other sources on Gaussian graphical models and exponential families.
The following result paraphrases \cite[Theorem 4.3]{SU}.

\begin{prop}\label{prop:decomposable}
If a linear covariance model $\mathcal L$ is given by zero restrictions on $\Sigma$, then 
the dual ML degree is equal to one if and only if the associated graph is chordal.
\end{prop}

It would be interesting to extend this result to  other combinatorial
families, such as colored covariance graph models (cf.~\cite{hojsgaard2008graphical}),
including structured Toeplitz matrices.

The following example illustrates Problem
\ref{prob:OrlandosProblem} and it raises some further questions.

\begin{exmp} \label{ex:onetwo}
We present a linear covariance model such that 
both the MLE and the dual MLE are rational functions.
Fix $n \geq 2$ and let $\mathcal{L}$ be the hyperplane
with equation $\sigma_{12} = 0$.
By Proposition~\ref{prop:decomposable}, the dual ML degree of $\mathcal{L}$ is one.
The model is dual to 
the decomposable undirected graphical model with missing edge $\{1,2\}$.

Following \cite{lauritzen, SU}, we obtain the
rational formula for its dual MLE:
\begin{equation} \label{eq:sigma12a}
\check k_{12}=W_{1,R}W_{R,R}^{-1}W_{R,2},\qquad 
{\rm and} \qquad \check k_{ij}= w_{ij}\quad\mbox{for }(i,j)\neq (1,2).
\end{equation}
Here $R = \{3,\ldots,n\}$ and $W_{\bullet,\bullet}$ 
is our notation for  submatrices of $W = (w_{ij}) = S^{-1}$.

The ML degree of the model $\mathcal{L}$ is also one. 
To see this, we note that
$\mathcal{L}$ is the DAG model with edges $i\to j$ whenever $i<j$ unless $(i,j)=(1,2)$.
By \cite[Section 5.4.1]{lauritzen}, the MLE of any Gaussian DAG model is rational.
In our case, we find $\, \hat K\;=\; W+A $,
where $A$ is the $n\times n$ matrix which is zero apart from the
upper left $2 \times 2$ block
$$
A_{12,12}\,\,\,=\,\,\,\begin{bmatrix}
	s_{11}^{-1} & 0\\
	0 & s_{22}^{-1}
	\end{bmatrix}\,\,-\,\,\frac{1}{s_{11}s_{22}-s_{12}^2}\begin{bmatrix}
	\phantom{-}s_{22} & -s_{12}\\
	-s_{12} & \phantom{-}s_{11}
\end{bmatrix}.
$$
The entries in $\check \Sigma = (\check K)^{-1}$ and $\hat \Sigma = (\hat K)^{-1}$ 
are rational functions in the data $s_{ij}$. But, 
unlike in the discrete case of \cite{DMS}, here the rational functions
are not products of linear forms. Problem \ref{prob:OrlandosProblem}
asks for an understanding of its irreducible factors.
\end{exmp}

Example \ref{ex:onetwo} raises many questions.
First of all, can we characterize all linear spaces $\mathcal{L}$ with rational formulas
for their MLE, or their dual MLE, or both of them?
Second, it would be interesting to study arbitrary models
$\mathcal{L}$ that are hyperplanes. Consider the entries for
$m = \binom{n+1}{2}-1$ in Tables \ref{table:ml-generic}  and \ref{tab:exp1}.
We know from \cite[Section 2.2]{SU} that
the dual ML degree equals $n-1$.
The ML degree seems to be $2^{n-1}-1$. In all cases there seems to be only one local
(and hence global) maximum.
How to prove these observations?
Finally, it is worthwhile to study the MLE when
$\mathcal{L}^{\perp}$ is a generic symmetric matrix of rank $r$.
What is the ML degree in terms of $r $ and $n$?

\section{Brownian Motion Tree Models}\label{sec7}

We now study
the linear space $\mathcal{L}_T$ associated with a rooted tree $T$
with~$n$ leaves. The equations of $\mathcal{L}_T$ are
$\sigma_{ij}=\sigma_{kl}$ whenever ${\rm lca}(i,j) = {\rm lca}(k,l)$.
 In the literature (cf.~\cite{felsenstein_maximum-likelihood_1973,
sturmfels2019brownian}) one assumes that the parameters $\theta_A$
in (\ref{eq:treedef}) are nonnegative. Here, we relax this hypothesis: we allow
all covariance matrices in the spectrahedron $\mathcal{L}_T \cap \mathbb{S}^n_+$.

The ML degree and its dual do not depend on how the leaves of a tree are labeled but only on the tree topology. For fixed $n$ each tree topology is uniquely identified by the set of clades.
 Since the root clade $\{1,\ldots,n\}$ and the leaf-clades $\{1\}$, \ldots,$\{n\}$ are part of every tree, they are omitted in our notation. For example, if $n=5$ then the tree $\{\{1,2\},\{3,4\},\{3,4,5\}\}$ is the binary tree with four inner vertices corresponding to the three non-trivial clades mentioned explicitly. This tree is depicted in Figure~\ref{fig:treesymm}.

We computed the ML degree and the dual ML degree of $\mathcal{L}_T$ for many trees $T$.
In Table~\ref{tab:BMT} we report results for five and six leaves. We notice that
the dual ML degree is exactly one for all binary trees.
  This suggests that the dual MLE is a rational function. Our
  main result in this section (Theorem~\ref{thm:bmtdual}) says that this is indeed true.

{\small
\begin{table}[]
\begin{tabular}{l|l|c|c}
n & Clades                                                                & ML degree & dual ML degree\\  \midrule
5 & \{1, 2, 3, 4\}                                             & 37  &  11     \\
5 & \{1, 2\}                                                & 37  &    11   \\
5 & \{1, 2, 3\}                                           & 81   & 16     \\
5 & \{1, 2\}, \{3, 4, 5\}                                     & 31 &  4      \\
5 & \{1, 2\}, \{3, 4\}                                        & 27 & 4       \\
5 & \{1, 2, 3\}, \{1, 2, 3, 4\}                               & 31 & 4       \\
5 & \{1, 2\}, \{1, 2, 3\}                                     & 31 & 4       \\
5 & \{1, 2\}, \{1, 2, 3, 4\}                                  & 27  & 4      \\
5 & \textbf{\{1, 2\}, \{3, 4\}, \{1, 2, 3, 4\}  }                      & 13 &  1      \\
5 & \textbf{\{1, 2\}, \{3, 4\}, \{1, 2, 5\} }                          & 17 & 1       \\
5 & \textbf{\{1, 2\}, \{1, 2, 3\}, \{1, 2, 3, 4\} }                   & 17 & 1       \\
\hline
6 & \{1, 2, 3, 4, 5\}                                          & 95  &  26     \\
6 & \{1, 2\}                                                  & 95  & 26      \\
6 & \{1, 2, 3, 4\}                                             & 259 &  44     \\
6 & \{1, 2, 3\}                                                & 259 & 44      \\
6 & \{1, 2, 3\}, \{4, 5, 6\}                                   & 221 &  16     \\
6 & \{1, 2\}, \{3, 4, 5, 6\}                                   & 101 &  11     \\
6 & \{1, 2, 3, 4\}, \{1, 2, 3, 4, 5\}                          & 101 &  11     \\
6 & \{1, 2\}, \{3, 4\}                                         & 81  &  11     \\
6 & \{1, 2\}, \{1, 2, 3\}                                      & 101 &  11     \\
6 & \{1, 2\}, \{3, 4, 5\}                                      & 181 & 16      \\
6 & \{1, 2\}, \{1, 2, 3, 4, 5\}                                & 81 &  11      \\
6 & \{1, 2, 3\}, \{1, 2, 3, 4\}                                & 221 & 16      \\
6 & \{1, 2, 3\}, \{1, 2, 3, 4, 5\}                             & 181 & 16      \\
6 & \{1, 2\}, \{1, 2, 3, 4\}                                   & 181 & 16      \\
6 & \{1, 2\}, \{3, 4\}, \{5, 6\}                               & 63 &  4      \\
6 & \{1, 2\}, \{3, 4\}, \{1, 2, 3, 4\}                         & 99 &  4      \\
6 & \{1, 2\}, \{1, 2, 3\}, \{4, 5, 6\}                         & 115 & 4      \\
6 & \{1, 2\}, \{3, 4, 5\}, \{3, 4, 5, 6\}                      & 115 & 4      \\
6 & \{1, 2\}, \{3, 4, 5\}, \{1, 2, 3, 4, 5\}                   & 99 &  4      \\
6 & \{1, 2\}, \{3, 4\}, \{1, 2, 5, 6\}                         & 83 & 4       \\
6 & \{1, 2\}, \{3, 4\}, \{1, 2, 3, 4, 5\}                      & 63 &  4      \\
6 & \{1, 2, 3\}, \{1, 2, 3, 4\}, \{1, 2, 3, 4, 5\}            & 115 & 4      \\
6 & \{1, 2\}, \{3, 4\}, \{1, 2, 5\}                           & 83  & 4      \\
6 & \{1, 2\}, \{1, 2, 3\}, \{1, 2, 3, 4\}                     & 115 & 4      \\
6 & \{1, 2\}, \{1, 2, 3\}, \{1, 2, 3, 4, 5\}                  & 83  & 4      \\
6 & \{1, 2\}, \{1, 2, 3, 4\}, \{1, 2, 3, 4, 5\}               & 83  & 4      \\
6 & \textbf{\{1, 2\}, \{3, 4\}, \{5, 6\}, \{1, 2, 3, 4\}  }             & 53  & 1      \\
6 & \textbf{\{1, 2\}, \{3, 4\}, \{1, 2, 5\}, \{3, 4, 6\}  }             & 61  & 1      \\
6 & \textbf{\{1, 2\}, \{3, 4\}, \{1, 2, 3, 4\}, \{1, 2, 3, 4, 5\} }    & 53 &  1      \\
6 &\textbf{\{1, 2\}, \{3, 4\}, \{1, 2, 5\}, \{1, 2, 5, 6\}  }        & 61  & 1      \\
6 & \textbf{\{1, 2\}, \{3, 4\}, \{1, 2, 5\}, \{1, 2, 3, 4, 5\} }      & 53  & 1      \\
6 & \textbf{\{1, 2\}, \{1, 2, 3\}, \{1, 2, 3, 4\}, \{1, 2, 3, 4, 5\} } & 61  & 1 \smallskip \\
\end{tabular}
\caption{ML degrees and dual ML degrees for 
Brownian motion tree models with five and six leaves.
 Binary trees are highlighted.}\label{tab:BMT}
\end{table}}

The equations (\ref{eq:opt_system2})
 for the dual ML degree can be written as $\,e_A^T (K-W) e_A=0$ for all clades $A$.
 Here $ W = (w_{ij})$ is given and $K^{-1} \in \mathcal{L}_T$ is unknown.
 We abbreviate
 \begin{equation}
 \label{eq:wnotation}
w_{A,B}\;=\;\sum_{i\in A}\sum_{j\in B} w_{ij}\;=\;e_A^T W e_B.
\end{equation}
The same notation is used for general matrices.
We present two examples with $n=4$.

\begin{exmp}
Consider the tree with clades $\{1,2\}$, $\{3,4\}$, shown in \cite[Figure 1]{sturmfels2019brownian}.
The dual MLE $\check K$ satisfies $\check k_{ii}= w_{ii}$ for $i=1,2,3,4$, 
and $\check k_{12}=w_{12}$, $\check k_{34}=w_{34}$, and
$$
\check k_{ij}\,=\,w_{12,34} \frac{w_{i,12}w_{j,34}}{w_{12,12}w_{34,34}} \qquad
\hbox{for $i\in \{1,2\}$, $j\in \{3,4\}$}.
$$
\end{exmp}

\begin{exmp}\label{ex:cater4}
The tree with clades $\{1,2\}$, $\{1,2,3\}$ has $\check k_{ii}=w_{ii}$, $\,\check k_{12}=w_{12}$, and
$$
\check k_{13}\,=\,w_{12,3}\frac{w_{1,12}}{w_{12,12}},\;\;
\check k_{14}\,=\,w_{123,4}\frac{w_{1,12}w_{12,123}}{w_{12,12}w_{123,123}},\;\;
\check k_{23}\,=\,w_{12,3}\frac{w_{2,12}}{w_{12,12}},
$$
$$
\check k_{24}\,=\,w_{123,4}\frac{w_{2,12}w_{12,123}}{w_{12,12}w_{123,123}},\;\;
\check k_{34}\,=\,w_{123,4}\frac{w_{123,3}}{w_{123,123}}.
$$
\end{exmp}
\noindent Both examples were computed in {\tt Mathematica} using the description of the Brownian motion
tree model in terms of the inverse covariance matrix  given in \cite{sturmfels2019brownian}.

Recall that for $v \in V$ we write $\de(v)$ for the set of leaves of $T$ that are descendants of $v$. The following theorem generalizes formulas in the above two examples.  It is our main result in
Section~\ref{sec7}.

\begin{thm}\label{thm:bmtdual}
Consider the model $\mathcal{L}_T$ given by
	a rooted binary tree $T$ with $n$ leaves.
	The dual MLE $\check K = (\check k_{ij})$ satisfies 
 $\check k_{ii}=w_{i,i}$ for all $i$, and its off-diagonal entries are
	\begin{equation} \label{eq:dualBMT} \qquad
	\check k_{ij}\;=\;w_{A,B}\prod_{u\to v}\frac{w_{\de(v),\de(u)}}{w_{\de(u),\de(u)}}
	\qquad \hbox{for $\,1\leq i<j\leq n$} .
	\end{equation}
	Here $A,B$ are the clades of the two children of $\,{\rm lca}(i,j)$.
The product is over all edges $u\to v$ of $\,T$,
 except for the two edges with $u = {\rm lca}(i,j)$,
on the path from $i$ to $j$ in~${T}$.
\end{thm}

\begin{proof}
We first rewrite (\ref{eq:opt_system2}) in terms of the coordinates $p_{ij}=-k_{ij}$ for $1\leq i<j\leq n$ and $p_{0i}=\sum_{j=1}^n k_{ij}$ for $1\leq i\leq n$. See \cite[Example 1.1]{sturmfels2019brownian} for an example of this coordinate change. The condition $K-W\in {{\mathcal L}_T}^\perp$ 
means that $k_{A,A}=w_{A,A}$ for every clade $A$ of $T$. This can be rewritten in the new coordinates as follows: 
\begin{itemize}
\item[(i)] $\sum_{j\neq i}p_{ij}=w_{i,i}$ for all $1\leq i\leq n$, and
\item[(ii)] $p_{A,B}=-w_{A,B}$ for all inner vertices $u$ of $T$, 
where $A|B$ is the partition of ${\rm clade}(u) = A \cup B$ given
by the two children of $u$.
\end{itemize}
Now fix $u$ with clade partition $A|B$ as above, so
$u={\rm lca}(i,j)$ for all $i \in A$, $j \in B$.
The parametrization $p_{0i}=1/t_i$ and  $p_{ij}=t_{{\rm lca}(i,j)}/(t_i t_j)$ 
 in \cite[Theorem~1.2]{sturmfels2019brownian} yields
$$ \begin{matrix}
p_{A,B}\,\,\,=\,\,\,t_{u} \cdot \sum_{i\in A}\frac{1}{t_i} \cdot \sum_{j\in B}\frac{1}{t_j}
\,\,\,=\,\,\, t_u \cdot p_{0,A} \cdot p_{0,B}. \end{matrix} $$
Using the equations in (ii), we obtain
\begin{equation}\label{eq:aux2}
	t_u \;=\;-\frac{w_{A,B}}{p_{0,A}p_{0,B}}\qquad\mbox{and hence }\qquad p_{ij}=-w_{A,B}\frac{p_{0i}}{p_{0,A}}\frac{p_{0j}}{p_{0,B}}.
\end{equation} 

We claim that the following identity holds for any clade $A\subseteq [n]$:
\begin{equation}\label{eq:p0A}
p_{0,A}\;=\;w_{[n],[n]}\prod_{u\to v}\frac{w_{\de(v),\de(u)}}{w_{\de(u),\de(u)}},	
\end{equation}
where the product is over all edges $u\to v$ of $T$ in the path from the root to the node
with clade $A$. Note that (\ref{eq:aux2}) and (\ref{eq:p0A}) imply
the formula (\ref{eq:dualBMT}) and so the theorem.

We now prove (\ref{eq:p0A}). Since $p_{0,[n]}=w_{[n],[n]}$, the claim holds for $A=[n]$.
Fix a clade $A\subset [n]$ and assume 
(\ref{eq:p0A})  for all clades $A_1\subset\cdots\subset A_k\subset A_{k+1}=[n]$ strictly containing $A_0=A$. For each $i=0,\ldots,k$ denote
$$
\alpha_i\;:=\;w_{A_{k+1},A_{k+1}}\frac{w_{A_k,A_{k+1}}}{w_{A_{k+1},A_{k+1}}}\cdots \frac{w_{A_i,A_{i+1}}}{w_{A_{i+1},A_{i+1}}}.
$$
By the induction hypothesis $p_{0,A_i}=\alpha_i$ for all $i=1,\ldots,k$.
Our goal is to prove that $p_{0,A}=\alpha_0$.
The clades $A_1\backslash A,\ldots,A_{k+1}\backslash A_k$ form a partition of $\bar A=[n]\backslash A$. 
We~have 
\begin{equation}\label{eq:p0A_identity}
	\begin{array}{rcl}
		p_{0,A}&=&w_{A,A}+k_{A,\bar A}\,\,= \,\,w_{AA}-p_{A,A_1\backslash A}-p_{A,A_2\backslash A_1}-\cdots-p_{A,A_{k+1}\backslash A_k} \\
		&=& w_{A,A}+w_{A,A_1\backslash A}+\sum_{i=1}^k w_{A_i,A_{i+1}\backslash A_i}\frac{p_{0A}}{p_{0,A_i}}.
	\end{array}
\end{equation}
Here the last equality follows because for every $i\in A$ and every
$j\in A_{l+1}\backslash A_{l}$ the vertex $u={\rm lca}(i,j)$ is the same. The clades of the children of $u$ are $A_{l} $ 
and $A_{l+1}\backslash A_{l}$.
Therefore, using (\ref{eq:aux2}),
 we get $p_{A,A_{l+1}\backslash A_l}=w_{A_l,A_{l+1}\backslash A_l}\frac{p_{0,A}}{p_{0,A_l}}$. 
We rewrite (\ref{eq:p0A_identity}) as
\begin{equation}\label{eq:aux}
w_{A,A_1}\;=\;p_{0,A}\left(1-\sum_{i=1}^k\frac{w_{A_i,A_{i+1}}-w_{A_i,A_i}}{\alpha_i} \right).	
\end{equation}
To simplify the bracketed expression, note that $\frac{w_{A_i,A_{i+1}}}{\alpha_i}=\frac{w_{A_{i+1},A_{i+1}}}{\alpha_{i+1}}$ and so, 
$$
1-\sum_{i=1}^k\frac{w_{A_i,A_{i+1}}-w_{A_i,A_i}}{\alpha_i}\;=\;1+\frac{w_{A_1,A_1}}{\alpha_1}-\frac{w_{A_k,A_{k+1}}}{\alpha_k}\;=\;\frac{w_{A_1,A_1}}{\alpha_1}.
$$
Plugging this back to (\ref{eq:aux}) gives
$$
p_{0,A}\;=\;\alpha_1\frac{w_{A,A_1}}{w_{A_1,A_1}}\;=\;\alpha_0.
$$
This proves the correctness of (\ref{eq:p0A}).

We have shown that  (\ref{eq:opt_system2}) implies
the rational formula (\ref{eq:dualBMT}) for $\check K$ in terms of $W$.
To argue that this is the MLE, we need that $W \in \mathbb{S}^n_+$ implies
$\check K \in \mathbb{S}^n_+$. For this, we use an analytic argument. Since
$W$ is positive definite, the dual likelihood function has a unique maximum $K=W$ over the 
whole cone $\mathbb S^n_+$. The model $\mathcal L_T\cap \mathbb S^n_+$ is a relatively closed subset of $\mathbb S^n_+$ and so the dual likelihood restricted to this set attains its maximum. The ML degree is equal to one and so there is at most one optimum in $\mathcal L_T$. We conclude there is exactly one optimum in $\mathcal L_T\cap \mathbb S^n_+$ and it is equal to $\check K$. \end{proof}

\begin{figure}
	\includegraphics[scale=.48]{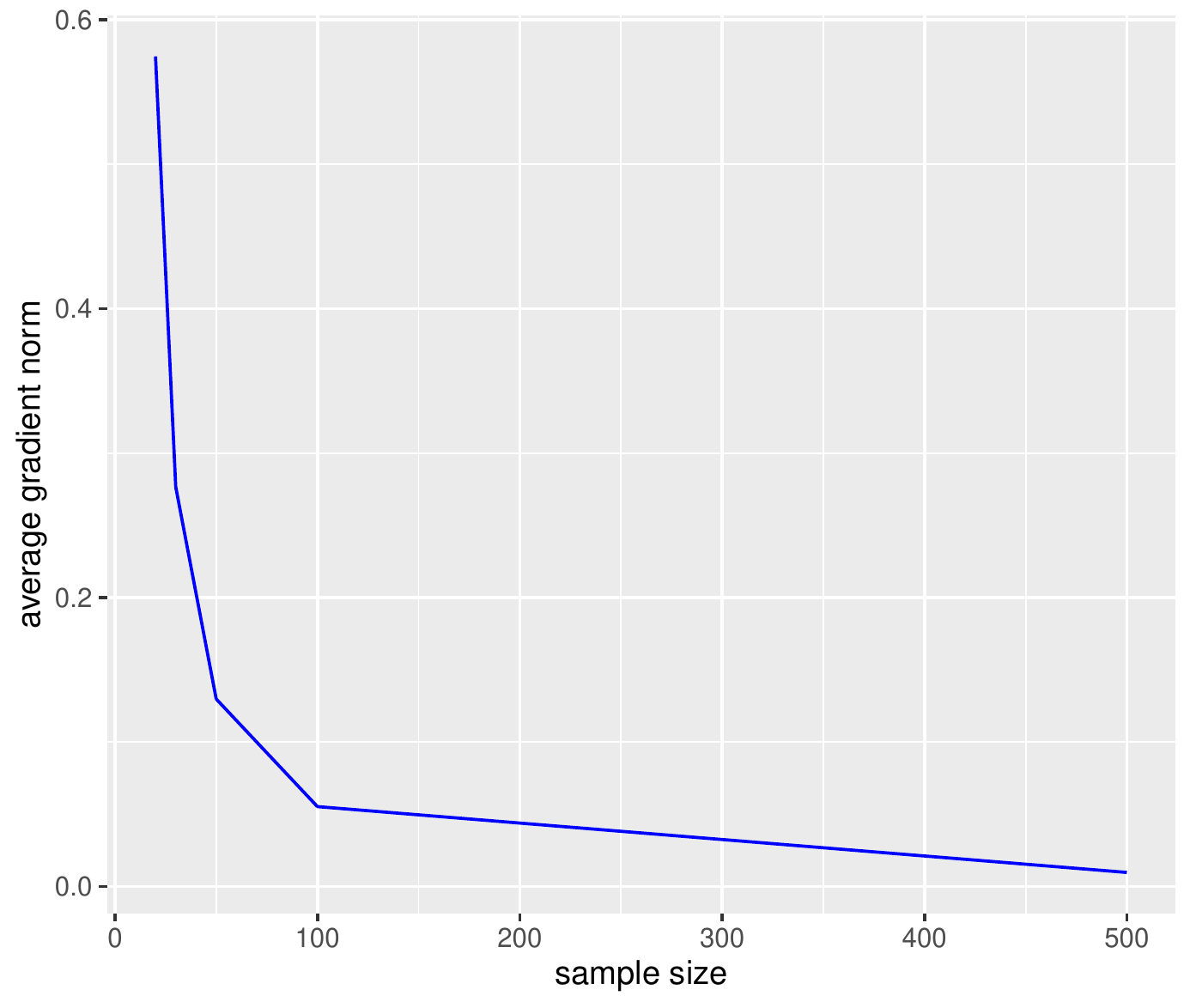} \quad \includegraphics[scale=.48]{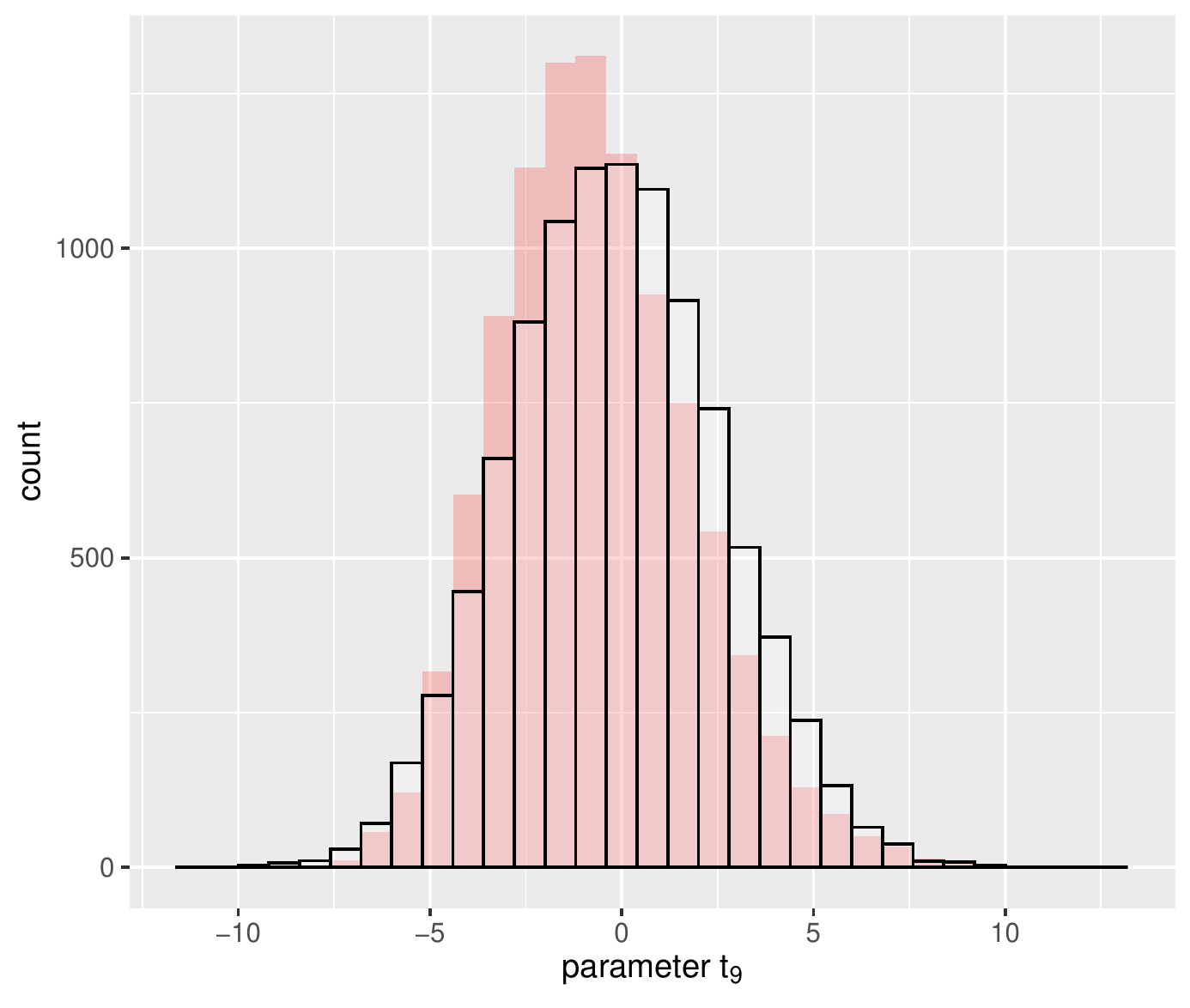}
	\caption{Left: The average value of the norm of the gradient of the likelihood function evaluated at the dual MLE for the sample size from 15 to 500. Right: Histogram of the estimated values of $\sqrt{N}(\check t_9-t_9)$ over 10000 iterations and for sample sizes 20 (red) and 1000 (white).}\label{fig:meangradient}
\end{figure}

Proposition~\ref{prop:dualstat} states that the dual MLE $\check K$ is asymptotically normal with the asymptotic variance as small as possible. The following experiment illustrates this.

\begin{exmp} Fix the five-leaf tree in Figure~\ref{fig:treesymm}, with clades ${\{1,2\},\{3,4\},\{3,4,5\}}$.
For simplicity assume that the data generating distribution has all parameters $\theta_A$
in (\ref{eq:treedef}) equal to one. 
The parameters of \cite[Theorem~1.2]{sturmfels2019brownian}, used in our proof above, are
	$$
	\bs t=(t_1,t_2,t_3,t_4,t_5,t_6,t_7,t_8,t_9)\;=\;(3,3,4,4,3,2,3,2,1).
$$

For each sample size $n=20,30,50,100,500$, we run $10000$ iterations to study the gradient of the likelihood function evaluated at the dual MLE given in Theorem~\ref{thm:bmtdual}. If the dual MLE is close to the MLE then we expect that the mean value of the norm of the gradient vector will be small. In Figure~\ref{fig:meangradient} the average norm for given sample sizes  is shown, and it is very close to zero even for moderate sample sizes. If $n=50$ this quantity is $0.13$ and for $n=500$ it is negligible and equal to $0.01$.	  

Asymptotic normality means that $\sqrt{N}(\check{\bs t}-\bs t)$ follows a normal distribution with mean zero and with covariance equal to the inverse of the Fisher information matrix; see \cite[Theorem 3.2]{christensen1989statistical}. Figure~\ref{fig:meangradient} 
shows the histogram of the estimates for the values $\sqrt{N}(\check t_9-t_9)$ 
for $N=20$ and $1000$.  Asymptotic normality is obvious  in this picture. 
\end{exmp}

The estimates in the previous example
were computed by evaluating the
function  in Theorem \ref{thm:bmtdual}.
We stress that nonlinear algebra and our software (\ref{eq:uurrll})
played an essential role in getting to this point. Namely, 
with computations as described in Section~\ref{sec5},
we created Table \ref{tab:BMT}. After seeing that
table, we conjectured that the dual MLE for 
binary trees is one. This led us to find the rational formula.
The expression (\ref{eq:dualBMT})
   is an alternating product of linear forms,
reminiscent of \cite[Theorem~1]{DMS}.
However, this structure does not generalize, by
Example \ref{ex:onetwo}, thus underscoring
Problem~\ref{prob:OrlandosProblem}.

\section*{Acknowledgements}
We thank Steffen Lauritzen for useful discussions
and for providing references.
We also thank
Jane Coons, Orlando Marigliano and Michael Ruddy
for their comments.
PZ was supported from the Spanish Government grants (RYC-2017-22544,PGC2018-101643-B-I00,SEV-2015-0563), and Ayudas Fundaci\'on BBVA a Equipos de Investigaci\'on Cientifica 2017.
ST was supported by the Deutsche Forschungsgemeinschaft (German Research Foundation) Graduiertenkolleg {\em Facets of Complexity} (GRK~2434).

\bigskip 

\bibliographystyle{alpha}
  \newcommand{\etalchar}[1]{$^{#1}$}

\end{document}